\documentclass[12pt]{article}%
\usepackage{amsfonts}
\usepackage{amsmath}
\usepackage{geometry}
\usepackage{amssymb}
\usepackage{graphicx}%
\providecommand{\U}[1]{\protect\rule{.1in}{.1in}}
\newtheorem{theorem}{Theorem}

\newtheorem{definition}[theorem]{Definition}
\newtheorem{example}[theorem]{Example}

\newtheorem{proposition}[theorem]{Proposition}
\newtheorem{remark}[theorem]{Remark}

\newenvironment{proof}[1][Proof]{\noindent\textbf{#1.} }{\ \rule{0.5em}{0.5em}}
\begin{document}

\title{A N\textsc{ew} A\textsc{pproach} \textsc{of} P\textsc{oint}
E\textsc{stimation} \textsc{from} T\textsc{runcated or Grouped and Censored}
D\textsc{ata}}
\author{A\textsc{hmed} G\textsc{uellil}$^{1}$ \textsc{and} T\textsc{ewfik}
K\textsc{ernane}$^{2}$\\$^{1}$\textit{Department of Probability and Statistics, Faculty of Mathematics
}\\\textit{\ University of Sciences and Technology USTHB,}\\\textit{BP 32 El-Alia, Algeria}\\$^{2}$\textit{Department of Mathematics, Faculty of Science}\\\textit{King Khalid University}\\\textit{Abha, Kingdom of Saudi Arabia}\\e-mail: guellilamed@yahoo.fr, \ tkernane@gmail.com}
\date{}
\maketitle

\begin{abstract}
We propose a new approach for estimating the parameters of a probability
distribution. It consists on combining two new methods of estimation. The
first is based on the definition of a new distance measuring the difference
between variations of two distributions on a finite number of points from
their support and on using this measure for estimation purposes by the method
of minimum distance. For the second method, given an empirical discrete
distribution, we build up an auxiliary discrete theoretical distribution
having the same support of the first and depending on the same parameters of
the parent distribution of the data from which the empirical distribution
emanated. We estimate then the parameters from the empirical distribution by
the usual statistical methods. In practice, we propose to compute the two
estimations, the second based on maximum likelihood principle of known
theoretical properties, and the first being as a control of the effectiveness
of the obtained estimation, and for which we prove the convergence in
probability, so we have also a criterion on the quality of the information
contained in the observations. We apply the approach to truncated or grouped
and censored data situations to give the flavour on the effectiveness of the
approach. We give also some interesting perspectives of the approach including
model selection from truncated data, estimation of the initial trial value in
the celebrate \textit{EM} algorithm in the case of truncation and merged
normal populations, a test of goodness of fit based on the new distance,
quality of estimations and data.

\end{abstract}


\section{Introduction}

Point estimation is the most popular forms of statistical inference (see
Lehmann and Casella \cite{Leh}). We introduce in this paper a new statistical
point estimation approach which found be useful in special practical
situations such as truncated and grouped and censored data. The data are said
to be truncated when measuring devices fail to report observations below
and/or above certain readings. For example, truncated data frequently arise in
the statistical analysis of astronomical observations ( see Efron and
Petrosian \cite{Efr}) and in medical data (see Klein and Zhang \cite{Kle}),
and if the truncation is ignored this can cause considerable bias in the
estimation. There exists in the literature many approaches of estimation from
"incomplete data" such as maximum likelihood based approach of the EM
algorithm (Hartley \cite{Hart}, Dempster et \textit{al }\cite{Demp}), or
nonparametric methods such as Kaplan-Meier (Kaplan and Meier \cite{Kap}) or
Lynden-Bell estimators (Lynden-Bell \cite{Lyn}). The purpose of the present
paper is to investigate another approach which consists on combining two new
methods of estimation and to apply it in the fixed type I censored or grouped
and censored data situations.

In the first method, we remark that in estimation problems we deal in general
with three functions: a theoretical probability law $f(\cdot,\theta)$ of a
random variable $X,$ depending on a parameter $\theta$ (real or vector
valued), an empirical distribution $\widehat{f}$ constructed from a sample of
observations drawn from the random variable $X,$ and an estimation
$\widetilde{f}$ (from an estimation $\widetilde{\theta}$ of $\theta$) obtained
through the empirical law $\widehat{f.}$ The empirical distribution
$\widehat{f}^{\text{ }}$is considered as a representative distribution of $f,$
but in practice it is reduced to only few of its characteristics such as the
mean and variance. The variational aspect of $\widehat{f}$ is often neglected
while its importance. We can easily find, for instance, two distributions
having the same support, mean and variance while their variations differ
significantly, or conversely having the same variations but their supports and
characteristic parameters are different. But two probability distributions
with same support and same variations in each subset of the support are
necessarily the same. We introduce then a new distance which measures the
difference between variations of two distributions on a finite number of
points and to use it for estimation purposes by the method of minimum
distance. Since the new measure is not equivalent to classical ones it will
give new insights that could not be investigated by classical distances.

In the second method, we remark that the empirical distribution arising from a
sample of observations can be viewed in fact as a conditional distribution as
it is built from the knowledge of the data. It will be then an estimation of
the theoretical conditional distribution with respect to the observations
before being an estimation for the parent distribution. This theoretical
conditional distribution is represented by the auxiliary distribution
introduced in this paper. To determine this distribution in discrete case, we
have simply to take the conditional distribution with respect to the observed
values and we proceed analogously for the continuous case. It should be noted
that in discrete case it is known as the truncated distribution which is the
conditional distribution given a truncation (see for example Shaw \cite{Sha})
but it is presented here in a general framework. We have to deal with two
discrete probability distributions having the same finite support, a
theoretical distribution and its empirical representation with respect to the
observations. The parameters of the former are those of the parent
distribution and the aim is to estimate them from the first instead of the
parent one as commonly used. We use classical tools such as the method of
moments or maximum likelihood principle. The setting that seems to us most
suitable for illustrating our approach is the one of truncated or grouped and
censored data. In usual practical problems, truncation can be on left or right
or in either situations, and the "cut off" can be deterministic or random. In
our approach, the truncation may be on any part of the range of the
distribution so that the setting is more general. Also, classical approaches
for truncated data are in general custom-made depending on specific problems
and distributions, or subjective based methods. Instead, our approach is quite
general and might be used in any situation where the underlying complete data
come from a known family of distributions. We confine ourselves as a first
presentation to fixed type I and grouped and censored data.

In the subsequent section, we propose a variational distance between
probability distributions. In Section 3, we define a truncation of data and
associated empirical and theoretical distributions and we use two different
methods for estimation from truncation, a first method using minimum of the
new distance introduced in this paper and a second method based on traditional
tools of estimations such as the method of maximum likelihood. In Section 4,
we present the new approach and we illustrate the procedure by three examples:
a binomial probability law, a normal distribution and a Gamma density
function. We present also a basic feature of the new approach which prove the
accuracy of the method and some illustrative examples. In Section 5, we give
some elements of comparison with the classical approach of estimation. In
Section 6, we list some perspectives of the new approach: model selection from
truncated data using the new distance, estimation of the first trial value in
the celebrate \textit{EM} algorithm for incomplete data in the case of
truncation and merged normal distributions, a goodness of fit test based on
the new distance, decision making about the quality of estimations and data.
Finally, concluding remarks are made some pointing to other possible
extensions and applications.

\section{A New Distance Between Probability Distributions}

As is usual, given a sample of $n$ independent and identically distributed
observations, $\left(  x_{1},...,x_{n}\right)  ,$ drawn from an unknown
discrete random variable $X$ falling in a discrete family of probability laws
$\mathcal{P=}\left\{  f(\cdot,\theta),\theta\in\mathbb{R}^{r}\right\}  $
depending on a parameter $\theta$ (real or vector valued), i.e.,
$f(x,\theta)=P\mathbb{(}X=x),$ one can summarize the sample into $k$ couples
$(y_{1},\widehat{f}_{1}),...,(y_{k},\widehat{f}_{k}),$ $k\leq n,$ where the
$y_{i}$ are the different values taken by the sample and $\widehat{f}$ is the
empirical law $\widehat{f}_{j}=n_{j}/n,$ where $n_{j}$ represents the absolute
frequency of the value $y_{j},$ $j=1,...,k.$

Usually, it is hoped that $\widehat{f}_{j}\approx f(y_{j},\theta),$ in a
certain probabilistic sense. But if the empirical distribution arises from
truncated data, we do not hope in general having $\widehat{f}(x)\approx
f(x,\theta),$ for the values $x$ in the support of $\widehat{f},$ since the
complete sample size $n$ is usually not reported. However, we expect
reasonably to have approximately%
\begin{equation}
\frac{\widehat{f}(x)}{\widehat{f}(y)}\approx\frac{f(x,\theta)}{f(y,\theta)},
\end{equation}
for any points in its support, only if the sample has serious irregularities.

Introduce the following \textit{distance} \textit{of} \textit{proportional
variations} between $f(\cdot,\theta)$ and $\widehat{f}$%
\begin{equation}
d_{v}(\widehat{f},f(\cdot,\theta))=\sum_{i,j\in\left\{  1,...,k\right\}
}\left\vert \frac{\widehat{f}_{i}}{\widehat{f}_{j}}-\frac{f(y_{i},\theta
)}{f(y_{j},\theta)}\right\vert . \label{dv1}%
\end{equation}
It turns out that this new distance, as we will show, measures the variations
between probability distributions.

In continuous case also, any sample $x_{1},...,x_{n}$ is summarized into $k$
couples $(y_{1},\widehat{f}_{1}),...$ $,(y_{k},\widehat{f}_{k}),$ $k\leq n$.
This can be done uniquely, by grouping for example the sample in classes where
the $y_{i}$ are the mid-classes (or class means) and $\widehat{f}_{i}%
=\widehat{f}(y_{i})$ where $\widehat{f}$ is an empirical density estimator, or
the data is presented in a grouped and censored form. The proportional
variational distance $d_{v}$ in this case, between the density $f(x,\theta)$
of $X$ and its empirical law $\widehat{f},$ is thus defined as (\ref{dv1}).
One of its main powerful feature is that when using traditional distances we
have to use the sample size $n$ through the expression of $\widehat{f}%
_{i}=n_{i}/(nh_{n}),$ where $h_{n}$ is the size of class intervals; but
sometimes, as for truncated data situations where measuring devices fail to
report even the number of sample points in certain ranges, then the real size
$n$ is not known, but a truncated sample size $n_{t}$ is instead used. Using
the ratios $\widehat{f}_{i}/\widehat{f}_{j}$ will clear up the effect of the
truncated sample size which can lead to considerable bias in the estimation.

Note that $d_{v}$ possesses the properties of symmetry and triangle
inequality. But in the identity property $d_{v}(f,g)(x,y)=0\Longleftrightarrow
f\equiv g,$ the equality between $f$ and $g$ must be understood in the sense
that $f$ and $g$ have the same variations on the points $x$ and $y.$ It should
be stressed that this new measure is not equivalent to classical ones and
should then give new insights and information about other characteristics and
features of probability distributions.

From now on $f$ shall represent a theoretical probability law in both discrete
or continuous cases and $\widehat{f}$ shall represent the corresponding
empirical law in both cases. Denote by $\Omega=\left\{  x\in\mathbb{R}%
,f(x,\theta)>0\right\}  $ the set of \textit{atoms} of $f$ or \textit{support}%
. Let $\mathcal{F}$ be the $\sigma-$algebra generated by sets $A=B\cap\omega$
where the $\omega$ are the Borel sets of $\mathbb{R}$ and $B\subset\Omega.$
For all $A\in\mathcal{F},$ we have $P\left(  A\right)  =\int_{A}f(x,\theta
)\mu(dx),$where $\mu$ is the Lebesgue measure on $\mathbb{R}$. In discrete
case, we have $P\left(  A\right)  =\sum_{x\in A}f(x,\theta).$

For all $i\geq1,$ we set $\Omega_{i}=\Omega,$ $\mathcal{F}_{i}=\mathcal{F}$
and $P_{i}=P.$ Let $\Omega^{n}=\Omega_{1}\times...\times\Omega_{n},$
$\mathcal{F}^{(n)}=\mathcal{F}_{1}\otimes...\otimes\mathcal{F}_{n}$ and
$P^{(n)}=P_{1}\otimes...\otimes P_{n}.$ The probability space $\left(
\Omega^{n},\mathcal{F}^{(n)},P^{(n)}\right)  $ represents the space of samples
of size $n$ from the random variable $X.$ We omit the subscript $n$ in
$\left(  \Omega^{n},\mathcal{F}^{(n)},P^{(n)}\right)  $ for notational
convenience and shall denote the sample space as $\left(  \Omega
,\mathcal{F},P\right)  .$

\subsection{A Notion of Variation between probability distributions}

We will discuss now the measure theoretic aspect of the new distance
introduced above. Let $P$ and $Q$ two probability measures defined on the same
measurable space $\left(  \Omega,\mathcal{F}\right)  $, $f$ and $g$ their
respective probability densities, not necessarily with respect to the same
measure and $E$ an event of this space. We say that $f$ and $g$ have the same
variation on $E$, if the respective restrictions of $f$ and $g$ on $E$, define
the same probability measure on $E$ endowed with the sigma algebra traces of
$\mathcal{F}$ on $E$.

\begin{definition}
Let $f$ and $g$ two probability distributions positive and defined on a part
$E$ not reduced to only one element. If in any point $(x,y)$ of $E\times E$,
we have:
\begin{equation}
\frac{f(x)}{f(y)}=\frac{g(x)}{g(y)}%
\end{equation}
then we say that $f$ and $g$ have same variations on $E$.
\end{definition}

\begin{example}
Let $f$ be a density of a probability measure $P$ and $E$ an event such that
$P(E)>0$. The restriction of $f$ on $E$ and the conditional distribution of
$f$ with respect to $E$ define the same probability measure on $E$ and
consequently they have the same variations on $E.$
\end{example}

\begin{definition}
Let $f$ and $g$ two probability distributions and $E$ an event on which they
are strictly positive. If $E$ is discrete and not reduced to only one element,
and one of the distributions $f$ and $g$ being discrete and the other may not
be discrete, we call distance in variations between $f$ and $g$ on $E$ the
quantity:%
\[
d_{v}(f,g)_{E}=\sum_{\left(  x,y\right)  \in E}\left\vert \frac{f(x)}%
{f(y)}-\frac{g(x)}{g(y)}\right\vert .
\]
If $E$ is an interval of $\mathbb{R}$ and, $f$ and $g$ are probability
densities on $\mathbb{R}$, with respect to Lebesgue measure $\mu$ on
$\mathbb{R}$, we call distance in variations between $f$ and $g$ on $E$, the
quantity:%
\[
d_{v}(f,g)_{E}=\iint\limits_{E\times E}\left\vert \frac{f(x)}{f(y)}%
-\frac{g(x)}{g(y)}\right\vert \mu(dx)\mu(dy).
\]

\end{definition}

Let be given a classical distance $d$ between two functions $f$ and $g$ which
associates for points $x$ and $y$ from the intersection of their domain of
definitions, the quantity $d\left(  f,g\right)  \left(  x,y\right)
=\left\vert f(x)-g(x)\right\vert +\left\vert f(y)-g(y)\right\vert .$

\begin{proposition}
We have the following properties for the distance $d_{v}:$\newline\textbf{1.}
$d(f,g)(x,y)=0\Longrightarrow d_{v}(f,g)(x,y)=0,$ the converse is not always
true.\newline\textbf{2.} Let $\widehat{f}$ be a kernel density estimation.
Then $\lim_{n\rightarrow\infty}d_{v}(\widehat{f},f)=0$ in probability.\newline%
\textbf{3. }Let $f$ and $g$ be two functions defined on $\mathbb{R}$ and
$E\subset\mathbb{R}$ satisfying:%
\[
\forall\left(  x,y\right)  \in E\times E,\text{ }d_{v}(f,g)(x,y)=0.
\]
If%
\[
\int_{\mathbb{R}}f\text{ }d\mu=\int_{\mathbb{R}}g\text{ }d\mu=1,
\]
where $\mu$ is the Lebesgue measure on $\mathbb{R}$, then%
\[
\mu\left(  \overline{E}\right)  =0\implies f=g\text{ \ \ }\mu-\text{almost
surely on }\mathbb{R}\text{.}%
\]

\end{proposition}

\begin{proof}
\textbf{1.} Follows directly from the definitions of $d$ and $d_{v}.$%
\newline\textbf{2. }Follows from the fact $\lim_{n\rightarrow\infty}%
d(\widehat{f},f)=0$ in probability (see Parzen \cite{Par}), then
$\lim_{n\rightarrow\infty}d_{v}(\widehat{f},f)=0$ in the same probabilistic
notion of convergence.\newline\textbf{3.} Fix $y_{0}\in E,$ we have
$f(x)/f(y_{0})=g(x)/g(y_{0})$ for all $x\in E.$ This implies that%
\[
\int_{E}f(x)dx=1\iff\int_{E}f(y_{0})\frac{g(x)}{g(y_{0})}dx=\frac{f(y_{0}%
)}{g(y_{0})}\int_{E}g(x)dx=1.
\]
We deduce that $f(y_{0})=g(y_{0}),$ and the result follows.
\end{proof}

\section{Truncated Data}

The truncated data specification, or generally \textit{incomplete data,
}implies the existence of two sample spaces $\mathcal{X}_{o}$ and
$\mathcal{X}_{t}$, such that the complete sample space is given by
$\Omega=\mathcal{X}_{o}\cup\mathcal{X}_{t}.$ The observed data $\mathbf{x}%
_{o}=\left(  x_{1},...,x_{n_{t}}\right)  ,$\ where $n_{t}$ is the truncated
sample size, are a realization from $\mathcal{X}_{o}$ and the unobserved data
$\mathbf{z=}\left(  x_{1}^{\ast},...,x_{n-n_{t}}^{\ast}\right)  ,$ where $n$
is the complete unknown sample size, are from $\mathcal{X}_{t}.$ The complete
data $\mathbf{x}=\mathbf{x}_{o}\cup\mathbf{z}$ is known only through the
observed data $\mathbf{x}_{o}$ (see Dempster, Laird and Rubin \cite{Demp} for
further explanations about incomplete data specification)\textit{. }

Consider a sample of observations $x_{1},...,x_{n}$ drawn from a theoretical
probability law $f(\cdot,\theta),$ depending on a parameter $\theta
\in\mathbb{R}^{r}.$ As usual, the data are summarized, in discrete or
continuous cases (as shown in Section 2), into $k$ couples $(y_{1},\widehat
{f}_{1}),...,(y_{k},\widehat{f}_{k}),$ $k\leq n,$ and let $\triangle=\left\{
u_{1},...,u_{m}\right\}  $ a part from the set $\left\{  y_{1},...,y_{k}%
\right\}  ,$ $m\leq k,$ which we will call \textit{truncation. }The observed
data is summarized by a truncation $\triangle_{o}=\left\{  u_{1}%
,...,u_{m}\right\}  $ and an empirical estimation $\widehat{f}_{o}$ and assume
that the unobserved data is also summarized by a set $\triangle_{t}%
=\{u_{1}^{\ast},...,u_{p}^{\ast}\}$ and $\widehat{f}_{t}.$

The structure of the new distance $d_{v}$ allows the following decomposition
property:%
\begin{align}
d_{v}(\widehat{f},f(\cdot,\theta))  &  =d_{v}(\widehat{f}_{o},f(\cdot
,\theta))+d_{v}(\widehat{f}_{t},f(\cdot,\theta))+\label{dec}\\
&  \sum_{\substack{u_{i}\in\triangle_{o}\\u_{j}^{\ast}\in\triangle_{t}%
}}\left\vert \frac{\widehat{f}_{o}\left(  u_{i}\right)  }{\widehat{f}%
_{t}\left(  u_{j}^{\ast}\right)  }-\frac{f(u_{i},\theta)}{f(u_{j}^{\ast
},\theta)}\right\vert +\sum_{\substack{u_{i}\in\triangle_{o}\\u_{j}^{\ast}%
\in\triangle_{t}}}\left\vert \frac{\widehat{f}_{t}\left(  u_{j}^{\ast}\right)
}{\widehat{f}_{o}\left(  u_{i}\right)  }-\frac{f(u_{j}^{\ast},\theta)}%
{f(u_{i},\theta)}\right\vert .\nonumber
\end{align}
The following proposition is typical for the new distance and is useful for
using the minimum of distance $d_{v}.$

\begin{proposition}
Let be given a truncated data $\triangle_{o}$ with corresponding empirical
estimation $\widehat{f}_{o}.$ Then $\lim_{n_{t}\rightarrow\infty}%
d_{v}(\widehat{f}_{o},f)=0$ in probability.
\end{proposition}

\begin{proof}
We have from Proposition 1 that $\lim_{n\rightarrow\infty}d_{v}(\widehat
{f},f)=0$ in probability. Then, from the decomposition property (\ref{dec}) we
obtain $\lim_{n\rightarrow\infty}d_{v}(\widehat{f}_{o},f)=\lim_{n_{t}%
\rightarrow\infty}d_{v}(\widehat{f}_{o},f)=0$ in probability.
\end{proof}

\subsection{An Auxiliary Distribution}

Define the empirical distribution $\widetilde{f}$ corresponding to a given
truncation $\triangle$ by:%
\[
\widetilde{f}(x)=\left\{
\begin{array}
[c]{c}%
\widetilde{f}_{i}\text{ \ \ \ \ if \ \ }x=u_{i},\text{ \ \ }i=1,...,m,\\
0\text{ \ \ \ \ \ \ \ \ \ \ \ \ \ \ \ otherwise, \ \ \ \ \ \ \ \ \ }%
\end{array}
\right.
\]
where the $\widetilde{f}_{i}$ satisfy the following set of proportional
allocation equations $\widetilde{f}_{i}/\widetilde{f}_{j}=\widehat{f}%
_{i}/\widehat{f}_{j},$ for $i,j=1,...,m$ and $\widetilde{f}_{1}+...+\widetilde
{f}_{m}=1.$

Define the following auxiliary distribution from $f(\cdot,\theta),$ which is
akin to the proportional allocation procedure for missing values (see Hartley
\cite{Hart}).%
\begin{equation}
h\left(  x,\theta\right)  =\left\{
\begin{array}
[c]{c}%
\dfrac{f(x,\theta)}{f(u_{1},\theta)+f(u_{2},\theta)+...+f(u_{m},\theta)}\text{
\ \ \ \ if \ \ }x=u_{i},\text{ \ \ }i=1,...,m,\\
0\text{ \ \ \ \ \ \ \ \ \ \ \ \ \ \ \ \ \ \ \ \ \ \ \ \ \ \ \ \ \ \ otherwise
\ }%
\end{array}
\right.  \label{auxil}%
\end{equation}

\begin{remark}
If the truncation is random, that is, there exists a random variable $T$ such
that we observe, for example, the random variable $X$ only if $X>T$ or $X<T,$
then the probability law used in (\ref{auxil}) is replaced by the conditional
law of $X$ with respect to $\left\{  X>T\right\}  $ or $\left\{  X<T\right\}
$ respectively.
\end{remark}

The auxiliary distribution $h$ was found be useful for estimation problems in
truncated data. Indeed, it is well known in classical estimation from
truncated data (see Hartley \cite{Hart}) that missing values could be
recovered by "proportional allocation" procedures, then the auxiliary
distribution $h,$ which is already based on proportional allocation, will be
an intuitive and natural tool for estimation purposes from truncated data. The
function $h$ is a theoretical probability distribution depending on the same
parameters of those of $f$. It has also the same \textit{support} as that of
$\widetilde{f}.$

\begin{definition}
We call $\widetilde{f}$ and $h(\cdot,\theta)$ the empirical and theoretical
distributions of a given truncation $\triangle=\left\{  u_{1},...,u_{m}%
\right\}  $ from a sample of observations $\left(  x_{1},...,x_{n}\right)  .$
\end{definition}

\section{The Approach of Estimation}

We will use mainly two methods of estimation. The first method is a minimum
distance estimation using the metric $d_{v}$ between the empirical and
theoretical distributions $\widehat{f}$ and $f(\cdot,\theta).$ The second is
similar to traditional ones such as the method of substitution or maximum
likelihood principle, by considering $\widetilde{f}$ as an empirical
estimation of $h(\cdot,\theta).$ The first is based on variational difference
between distributions and the second in the sense of an euclidean difference
and hence they treat different aspects of the sample of observations. If for a
given data they give different estimations, we cannot suspect the approaches
but we can say that the data do not restore in a coherent way all aspects of
the probability distribution from which it emanated. If on the other hand they
give significantly the same estimations we can assert that the estimation is
credible since through different aspects it has given the same distribution.
That is the distribution which fits the best the empirical distribution.
Practically, we propose to calculate the estimations by the two methods and
take the second one since based on maximum likelihood principle of good known
theoretical properties. We use then the first as a tool of decision on whether
the estimation is credible or not. The estimation will then be considered as
credible in cases where the two methods give approximately the same estimation.

\subsection{Convergence in Probability of the Minimum Distance Estimator}

Let $X_{1},X_{2},...,X_{n}$ a sample with $X_{i}\sim f(x,\theta),$
$\theta=\left(  \theta_{1},...,\theta_{s}\right)  ^{t}\in\Theta\subseteq
\mathbb{R}^{s},$ with%
\begin{equation}
f(x,\theta)=K(x)\times\exp\left\{  \sum_{k=1}^{s}\theta_{k}T_{k}%
(x)+A(\theta)\right\}  , \label{fam1}%
\end{equation}
$x\in\mathcal{X}\subseteq\mathbb{R},$ where $\mathcal{X}$ is a Borel set of
$\mathbb{R}$ such that $\mathcal{X=}\left\{  x:f(x,\theta)>0\right\}  $ for
all $\theta\in\Theta.$

The family (\ref{fam1}) is very rich, one finds there, for example, the family
of the normal laws, and the family of the laws of Poisson. We assume that the
support $\mathcal{X}$ does not depend on $\theta.$ Denote by $\widetilde
{\theta}_{n}$ the estimator by the minimum of metric $d_{v}$ between the
empirical and theoretical distributions $\widehat{f}_{n}$ (based on a sample
of size $n$) and $f(\cdot,\theta),$ that is%
\[
\widetilde{\theta}_{n}=\arg\min_{\theta}d_{v}(f(\cdot,\theta),\widehat{f}%
_{n}).
\]
This estimator falls into the class of M-estimators. Using well known theorems
on the convergence of M-estimators (see for example Amemiya \cite{Ame}) we
will prove that $\widetilde{\theta}_{n}$ converges in probability to the true parameter.

\begin{proposition}
Let $X_{1},X_{2},...,X_{n}$ be a sample from the family of distributions
(\ref{fam1}). If the set of natural parameters $\Theta$ is covex and the true
parameter $\theta$ is an interior point of $\Theta,$ then the estimator
$\widetilde{\theta}_{n}$ by the minimum of the distance of variations $d_{v}$
converges in probability to the true parameter $\theta,$ i.e.,%
\[
\widetilde{\theta}_{n}\overset{P}{\longrightarrow}\theta.
\]

\end{proposition}

\begin{proof}
Since we search for a minimum of the criterion function $d_{v},$ it suffices
to show, under the assumptions of the family (\ref{fam1}) and the convexity of
the set $\Theta,$ that $d_{v}(\theta,\underline{x})$ seen as a function of
$\theta$ is a convex function (see Amemiya \cite{Ame}). Hence, this reduces
the problem to the convexity of%
\[
\delta_{ij}(\theta)=\left\vert \frac{f(y_{i},\theta)}{f(y_{j},\theta)}%
-\frac{\widehat{f}(y_{i})}{\widehat{f}(y_{j})}\right\vert .
\]
For $\lambda,\mu\in\mathbb{R}$ with $\lambda+\mu=1$, and $\theta^{(1)}%
,\theta^{(2)}\in\Theta,$ we have%
\begin{equation}
\delta_{ij}(\lambda\theta^{(1)}+\mu\theta^{(2)})=\left\vert C_{ij}\exp\left\{
\sum_{k=1}^{s}\left[  \lambda\theta_{k}^{(1)}+\mu\theta_{k}^{(2)}\right]
\left(  T_{k}(y_{i})-T_{k}(y_{j})\right)  \right\}  -A_{ij}\right\vert
\end{equation}
where $C_{ij}=K(y_{i})/K(y_{j})$ and assume that $C_{ij}>0$ and $A_{ij}%
=\widehat{f}(y_{i})/\widehat{f}(y_{j}).$\newline we have from the convexity of
the exponential function that%
\begin{align*}
\exp\left\{  \sum_{k=1}^{s}\left[  \lambda\theta_{k}^{(1)}+\mu\theta_{k}%
^{(2)}\right]  \left(  T_{k}(y_{i})-T_{k}(y_{j})\right)  \right\}   &
\leq\lambda\exp\left\{  \sum_{k=1}^{s}\theta_{k}^{(1)}\left(  T_{k}%
(y_{i})-T_{k}(y_{j})\right)  \right\} \\
&  +\mu\exp\left\{  \sum_{k=1}^{s}\theta_{k}^{(2)}\left(  T_{k}(y_{i}%
)-T_{k}(y_{j})\right)  \right\}  ,
\end{align*}
then%
\[
C_{ij}\exp\left\{  \sum_{k=1}^{s}\left[  \lambda\theta_{k}^{(1)}+\mu\theta
_{k}^{(2)}\right]  \left(  T_{k}(y_{i})-T_{k}(y_{j})\right)  \right\}
-A_{ij}\leq
\]%
\[
\lambda C_{ij}\exp\left\{  \sum_{k=1}^{s}\theta_{k}^{(1)}\left(  T_{k}%
(y_{i})-T_{k}(y_{j})\right)  \right\}  +\mu C_{ij}\exp\left\{  \sum_{k=1}%
^{s}\theta_{k}^{(2)}\left(  T_{k}(y_{i})-T_{k}(y_{j})\right)  \right\}
\]%
\[
-\left(  \lambda+\mu\right)  A_{ij}\leq\lambda\left[  C_{ij}\exp\left\{
\sum_{k=1}^{s}\theta_{k}^{(1)}\left(  T_{k}(y_{i})-T_{k}(y_{j})\right)
\right\}  -A_{ij}\right]  +
\]%
\[
\mu\left[  C_{ij}\exp\left\{  \sum_{k=1}^{s}\theta_{k}^{(2)}\left(
T_{k}(y_{i})-T_{k}(y_{j})\right)  \right\}  -A_{ij}\right]  .
\]
Introducing the absolute value we get%
\[
\delta_{ij}(\lambda\theta^{(1)}+\mu\theta^{(2)})=\left\vert C_{ij}\exp\left\{
\sum_{k=1}^{s}\left[  \lambda\theta_{k}^{(1)}+\mu\theta_{k}^{(2)}\right]
\left(  T_{k}(y_{i})-T_{k}(y_{j})\right)  \right\}  -\left(  \lambda
+\mu\right)  A_{ij}\right\vert
\]%
\[
\leq\lambda\left\vert C_{ij}\exp\left\{  \sum_{k=1}^{s}\theta_{k}^{(1)}\left(
T_{k}(y_{i})-T_{k}(y_{j})\right)  \right\}  -A_{ij}\right\vert
\]%
\[
+\mu\left\vert C_{ij}\exp\left\{  \sum_{k=1}^{s}\theta_{k}^{(2)}\left(
T_{k}(y_{i})-T_{k}(y_{j})\right)  \right\}  -A_{ij}\right\vert =\lambda
\delta_{ij}(\theta^{(1)})+\mu\delta_{ij}(\theta^{(2)}).
\]
Hence $\delta_{ij}(\theta)$ is a convex function of $\theta,$ which implies
the convexity of $d_{v}(\theta,\underline{x})$ seen as a function of $\theta$
and then the convergence in probability of the minimum of distance $d_{v}$ estimator.
\end{proof}

\subsection{A Maximum Likelihood Principle with the Auxiliary Distribution}

We firstly begin in a general situation, that of the one-parameter exponential
family, to show how to use the procedure explained below in the case of the
new method. Consider the one-parameter exponential family with density%
\begin{equation}
f(x,\theta)=K(x)\times\exp[\theta T(x)-A(\theta)],
\end{equation}
where $\theta$ is the parameter, $T$ a statistic, $K(x)$ a function of $x$ and
$A$ is a function of the parameter $\theta$. Let us use the maximum likelihood
principle. Consider a sample of observations $x_{1},...,x_{n}$ from which we
derive the support $\triangle=\left\{  y_{1},...,y_{k}\right\}  .$ We then
construct the auxiliary distribution from the support $\triangle$, expressed
in the following form%
\begin{equation}
h(x,\theta)=\frac{K(x)\times\exp[\theta T(x)-A(\theta)]}{\sum_{i=1}^{k}%
K(y_{i})\times\exp[\theta T(y_{i})-A(\theta)]}.
\end{equation}
We have to maximize the likelihood function given in our case by%
\begin{equation}
\mathbf{L}_{h}\left(  y,\theta\right)  =%
{\displaystyle\prod\limits_{i=1}^{k}}
h(y_{i},\theta).
\end{equation}
Without loss of generality, we assume that the class intervals are the same.
Then, we have%
\begin{equation}
\log\mathbf{L}_{h}\left(  y,\theta\right)  =%
{\displaystyle\sum\limits_{i=1}^{k}}
\log h(y_{i},\theta)=%
{\displaystyle\sum\limits_{i=1}^{k}}
\frac{n_{i}}{n}\log\left[  \frac{K(y_{i})\times\exp[\theta T(y_{i}%
)-A(\theta)]}{%
{\displaystyle\sum\limits_{i=1}^{k}}
K(y_{i})\times\exp[\theta T(y_{i})-A(\theta)]}\right]  ,
\end{equation}
taking the derivative and solving the score equation on $\theta$ we obtain an
estimator of the parameter $\theta$ satisfying the relation%
\begin{equation}%
{\displaystyle\sum\limits_{i=1}^{k}}
\frac{n_{i}}{n}\ \frac{%
{\displaystyle\sum\limits_{i=1}^{k}}
T(y_{i})\times f(y_{i},\theta)}{%
{\displaystyle\sum\limits_{i=1}^{k}}
f(y_{i},\theta)}=%
{\displaystyle\sum\limits_{i=1}^{k}}
\frac{n_{i}}{n}T(y_{i}). \label{famexpo2}%
\end{equation}
The later result may be obtained directly by the method of moments, but we
have presented the maximum likelihood method since it is widely used in
statistical inference.

In order to test the performance of the proposed approach, we use synthetic
data sets which were generated by simulation from three examples of
probability law: binomial law, normal density and a Gamma distribution. The
examples were selected from various simulation studies from different family
of probability distributions and the two methods have shown their
effectiveness and never deviate significantly from the true parameter. The
reason for using synthetic data sets is that the true parameters for the
synthetic datasets are known and the accuracy of results obtained by using the
two new methods can be compared.

\subsection{Examples}

\textbf{Binomial distribution. }We generated a synthetic data set of size
$500$ from a binomial law $\mathcal{B}(n,p)$ with $n=10$ and $p=0.3,$ and
denote by $f(y;p)=C_{n}^{y}p^{y}(1-p)^{n-y}$ its probability mass function.
The data are summarized in the following table.

\begin{center}
\textbf{Table 1. \smallskip}%

\begin{tabular}
[c]{ccccccccc}\hline
$y_{i}$ & $0$ & $1$ & $2$ & $3$ & $4$ & $5$ & $6$ & $7$\\\hline
$n_{i}$ & $15$ & $71$ & $108$ & $134$ & $97$ & $47$ & $23$ & $5$\\\hline
\end{tabular}

\end{center}

Our aim is to estimate the parameter $p,$ with the knowledge of $n=10$, from
different truncation of data.

For illustrating the two methods, consider the truncation $\triangle=\left\{
2,3,4,5\right\}  $ with truncated sample size $n_{t}=386.$ We have then a
truncation proportion of $Q=100(n-n_{t})/n=22,8$ $\%$ in data. For the first
method, we have to search the value of the parameter $p$ which minimizes the
distance $d_{v},$ that is:%
\[
\min_{p}d_{v}(\widehat{f},f)=\min_{p}\sum\limits_{\substack{i,j\in
\triangle\\i\neq j}}\left\vert \frac{f(y_{i};p)}{f(y_{j};p)}-\frac{n_{i}%
}{n_{j}}\right\vert ,
\]

Using computer algebra package, we obtain the result $\widetilde{p}%
_{1}=0.299.$

For the second method, the empirical distribution $\widetilde{f}$ given the
truncation $\Delta=\left\{  2,3,4,5\right\}  $ is given by $\widetilde
{f}(2)=108/386,$ $\widetilde{f}(3)=134/386$, $\widetilde{f}(4)=97/386,$
$\widetilde{f}(5)=47/386$ and $\widetilde{f}(x)=0$ if $x\notin\Delta.$

The auxiliary distribution $h(\cdot,p)$ is given by:%
\begin{equation}
h(x,p)=%
\genfrac{\{}{.}{0pt}{}{\dfrac{f(x,p)}{f(2,p)+f(3,p)+f(4,p)+f(5,p)}\text{
\ \ \ \ if }x=u_{i},\text{ \ \ }u_{i}\in\left\{  2,3,4,5\right\}  }{0\text{
\ \ \ \ \ \ \ \ \ \ \ \ \ \ \ \ \ \ \ \ \ \ \ \ otherwise.}}%
\end{equation}
By the method of substitution, the estimation of $p$ is obtained by solving
the equation:%
\begin{equation}
\sum_{u_{i}\in\left\{  2,3,4,5\right\}  }u_{i}\times h(u_{i},p)=\sum_{u_{i}%
\in\left\{  2,3,4,5\right\}  }u_{i}\times\widetilde{f}(u_{i})
\end{equation}
Using a computer algebra package we obtain the result $\widetilde{p}_{2}=0.3$.

In the following table we present the estimations $\widetilde{p}_{1}$ from the
first method using minimum distance approach using the distance $d_{v}$, and
$\widetilde{p}_{2}$ from the auxiliary distribution, of the parameter $p,$ for
known $n,$ according to the truncation $\triangle=\left\{  u_{1}%
,...,u_{m}\right\}  $ considered.\newpage

\begin{center}
\textbf{Table 2.} The estimations $\widetilde{p}_{1}$ and $\widetilde{p}_{2}$
by the new approach of the parameter $p$

of the binomial probability law $\mathcal{B}(n,p)$ with $p=0.3$ and known
$n=10$.\smallskip%

\begin{tabular}
[c]{cccccc}\hline
&  & Truncated & Proportion of &  & \\
$n%
{{}^\circ}%
$ & $\triangle$ & sample size $n_{t}$ & truncation $Q$ $(\%)$ & $\widetilde
{p}_{1}$ & $\widetilde{p}_{2}$\\\hline
$1$ & $\left\{  0,1,2,3,4,5,6,7\right\}  $ & $500$ & $0$ & $0.305$ & $0.298$\\
$2$ & $\left\{  0,1,2,3,4,5\right\}  $ & $472$ & $5.6$ & $0.295$ & $0.293$\\
$3$ & $\left\{  1,2,3,4,5\right\}  $ & $457$ & $8.6$ & $0.288$ & $0.292$\\
$4$ & $\{0,1,2,3,4\}$ & $425$ & $15$ & $0.295$ & $0.293$\\
$5$ & $\left\{  1,2,3,4\right\}  $ & $410$ & $18$ & $0.287$ & $0.292$\\
$6$ & $\left\{  0,2,3,4,5\right\}  $ & $401$ & $19.8$ & $0.295$ & $0.298$\\
$7$ & $\left\{  2,3,4,5\right\}  $ & $386$ & $22.8$ & $0.299$ & $0.3$\\
$8$ & $\left\{  0,1,3,4,5\right\}  $ & $364$ & $27.2$ & $0.295$ & $0.289$\\
$9$ & $\{0,2,3,4\}$ & $354$ & $29.2$ & $0.295$ & $0.301$\\
$10$ & $\{1,3,4,5\}$ & $349$ & $30.2$ & $0.287$ & $0.287$\\
$11$ & $\left\{  2,3,4\right\}  $ & $339$ & $32.2$ & $0.305$ & $0.305$\\
$12$ & $\{0,3,4,5\}$ & $293$ & $41.4$ & $0.295$ & $0.293$\\
$13$ & $\left\{  2,4,5,6,7\right\}  $ & $280$ & $44$ & $0.308$ & $0.307$\\
$14$ & $\{0,1,2,5,6,7\}$ & $269$ & $46.2$ & $0.298$ & $0.299$\\
$15$ & $\left\{  0,1,4,5,6,7\right\}  $ & $258$ & $48.4$ & $0.3013$ &
$0.295$\\
$16$ & $\left\{  0,4,5,6,7\right\}  $ & $187$ & $62.6$ & $0.3071$ & $0.302$\\
$17$ & $\left\{  0,5,6,7\right\}  $ & $90$ & $82$ & $0.3014$ & $0.301$\\
$18$ & $\left\{  0,5\right\}  $ & $62$ & $87.6$ & $0.2937$ & $0.294$\\\hline
\end{tabular}

\end{center}

As previously said, the two estimations by the new approach, $\widetilde
{p}_{1}$ and $\widetilde{p}_{2},$ are accurate in all cases and close to each
other. Furthermore, the truncation proportion has no effect on the quality of
estimations. The two estimations are also not sensitive to small cell
probabilities as for truncations including the value $y_{8}=7.$ It should be
noted that the classical estimation by maximum likelihood without truncation
is $\widehat{p}=0.297,$ and considering our approach we obtained the
estimations $\widetilde{p}_{1}=0.3053$ for the first method and $\widetilde
{p}_{2}=0.2978$ for the second.\medskip\newline\textbf{Normal distribution.
}Consider a sample of size $400$ drawn from a normal population with mean
$m=0$ and standard deviation $\sigma=1.$ Consider the data falling in $11$
fixed class intervals as shown in the following table, with mid-classes
$u_{i}$ and absolute frequencies $n_{i}$\newline\textbf{Table 3.}%
\smallskip\newline$%
\begin{tabular}
[c]{cccccccccccc}\hline
$y_{i}$ & $-2.581$ & $-2.06$ & $-1.533$ & $-1.009$ & $-0.485$ & $0.039$ &
$0.563$ & $1.086$ & $1.610$ & $2.134$ & $2.658$\\\hline
$n_{i}$ & $5$ & $8$ & $23$ & $48$ & $71$ & $89$ & $72$ & $43$ & $25$ & $10$ &
$6$\\\hline
\end{tabular}
\medskip$

The number of bins can be selected from an optimal procedure developed by
Birg\'{e} and Rozenholc \cite{Bir1}. Let the following table where we estimate
simultaneously $m$ and $\sigma$ by the minimum distance procedure with $dv$.
We denote the estimations by $\tilde{m}_{1}$ and $\tilde{\sigma}_{1}.$ In each
line of the table the estimates are made starting from the table of
frequencies based on the observations indicated in the first column. The
truncated sample size is denoted by $n_{t}.$ We have then a truncation
proportion of $Q=100(n-n_{t})/n$ in data.

\begin{center}
\textbf{Table 4.}\smallskip\newline%

\begin{tabular}
[c]{ccccc}\hline
$S$ & $n_{t}$ & $Q\%$ & $\widetilde{m}_{1}$ & $\tilde{\sigma}_{1}$\\\hline
$\left\{  y_{1},y_{2},y_{3},y_{4},y_{5},y_{6},y_{7},y_{8},y_{9},y_{10}%
,y_{11}\right\}  $ & $400$ & $0$ & $0.083$ & $1.130$\\
$\left\{  y_{1},y_{2},y_{3},y_{4},y_{5},y_{6},y_{7},y_{8},y_{9}\right\}  $ &
$384$ & $4$ & $0.003$ & $1.092$\\
$\left\{  y_{2},y_{3},y_{4},y_{5},y_{6},y_{7},y_{8},y_{9}\right\}  $ & $379$ &
$5.25$ & $0.054$ & $0.977$\\
$\left\{  y_{3},y_{4},y_{5},y_{6},y_{7},y_{8},y_{9}\right\}  $ & $371$ &
$7.25$ & $0.052$ & $0.993$\\
$\left\{  y_{4},y_{5},y_{6},y_{7},y_{8},y_{9}\right\}  $ & $348$ & $13$ &
$0.043$ & $1.017$\\
$\left\{  y_{5},y_{6},y_{7},y_{8},y_{9}\right\}  $ & $300$ & $25$ & $0.052$ &
$1.012$\\
$\left\{  y_{3},y_{4},y_{5},y_{6}\right\}  $ & $231$ & $42.25$ & $0.303$ &
$1.104$\\
$\left\{  y_{6},y_{7},y_{8},y_{9}\right\}  $ & $229$ & $42.75$ & $-0.225$ &
$1.140$\\
$\left\{  y_{6},y_{7},y_{8}\right\}  $ & $204$ & $49$ & $-0.065$ & $1.052$\\
$\left\{  y_{3},y_{5},y_{7}\right\}  $ & $166$ & $58.5$ & $0.052$ & $0.993$\\
$\left\{  y_{2},y_{3},y_{4},y_{5}\right\}  $ & $150$ & $62.5$ & $-0.137$ &
$0.904$\\
$\left\{  y_{3},y_{4},y_{5}\right\}  $ & $142$ & $64.5$ & $-0.151$ &
$0.893$\\\hline
\end{tabular}

\end{center}

\begin{remark}
In practice, the bins are in fact chosen after obtaining the truncated sample
so the results should be more efficient, but this does not affect the
preceding results obtained after grouping the whole sample and truncate from
the bins since the aim is to give some feel about the accuracy of the
estimations. Also we can avoid grouping the observations by considering
empirical frequencies obtained from kernel density estimations.
\end{remark}

\subsubsection{Gamma probability density}

Consider a sample of size $800$ drawn from a Gamma distribution $G(a,b)$ with
density given by%
\begin{equation}
f(x\mid a,b)=\frac{1}{b^{a}\Gamma(a)}x^{a-1}\exp\left(  -\frac{x}{b}\right)
,\text{ \ \ }x\geq0,
\end{equation}
and parameters $a=7$ and $b=3.$ Consider the data falling in $16$ fixed class
intervals as shown in the following table, with mid-classes $u_{i}$ and
absolute frequencies $n_{i}:$

\textbf{Table 5.}\smallskip

$%
\begin{tabular}
[c]{ccccccccccc}\hline
$u_{i}$ & $5.89$ & $8.72$ & $11.56$ & $14.39$ & $17.23$ & $20.06$ & $22.89$ &
$25.73$ & $28.56$ & $31.39$\\\hline
$n_{i}$ & $11$ & $40$ & $60$ & $108$ & $118$ & $104$ & $100$ & $74$ & $63$ &
$53$\\\hline
\end{tabular}
\ $%

\begin{tabular}
[c]{cccccc}\hline
$34.23$ & $37.06$ & $39.89$ & $42.73$ & $45.56$ & $48.39$\\\hline
$27$ & $21$ & $11$ & $5$ & $3$ & $2$\\\hline
\end{tabular}
\medskip

In the following table we show the estimations $\widetilde{b}_{1}$ from the
minimum of distance $d_{v}$ and $\widetilde{b}_{2}$ by the second method for
the parameter $b,$ with known $a=10,$ according to the truncation $\triangle$
considered.\newpage

\begin{center}
\textbf{Table 6.} The estimations $\widetilde{b}_{1}$ and $\widetilde{b}_{2}$
by the new approach of the parameter $b$

of the Gamma probability distribution $G(a,b)$ with $b=3$ and known
$a=7$.\smallskip%

\begin{tabular}
[c]{cccccc}\hline
$n%
{{}^\circ}%
$ & $\triangle$ & $n_{t}$ & $Q$ $(\%)$ & $\widetilde{b}_{1}$ & $\widetilde
{b}_{2}$\\\hline
$1$ & $\left\{  u_{1},u_{2},u_{3},u_{4},u_{5},u_{6},u_{7},u_{8},u_{9}\right.
$ & $800$ & $0$ & $3.018$ & $3.054$\\
& $\left.  u_{10},u_{11},u_{12},u_{13},u_{14},u_{15},u_{16}\right\}  $ &  &  &
& \\
$2$ & $\left\{  u_{2},u_{3},u_{4},u_{5},u_{6},u_{7},u_{8},u_{9},\right.  $ &
$787$ & $1.625$ & $2.980$ & $3.065$\\
& $\left.  u_{10},u_{11},u_{12},u_{13},u_{14},u_{15}\right\}  $ &  &  &  & \\
$3$ & $\left\{  u_{1},u_{2},u_{3},u_{4},u_{5},u_{6},u_{7},u_{8},u_{9}%
,u_{10},u_{11},u_{12}\right\}  $ & $779$ & $2.625$ & $3.012$ & $3.068$\\
$4$ & $\left\{  u_{1},u_{2},u_{3},u_{4},u_{5},u_{6},u_{7},u_{8},u_{9}%
,u_{10}\right\}  $ & $731$ & $8.625$ & $2.895$ & $3.059$\\
$5$ & $\left\{  u_{2},u_{3},u_{4},u_{5},u_{6},u_{7},u_{8},u_{9},u_{10}%
\right\}  $ & $720$ & $10$ & $3.063$ & $3.075$\\
$6$ & $\left\{  u_{3},u_{4},u_{5},u_{6},u_{7},u_{8},u_{9},u_{10}\right\}  $ &
$680$ & $15$ & $3.157$ & $3.119$\\
$7$ & $\left\{  u_{1},u_{2},u_{3},u_{4},u_{5},u_{6},u_{7},u_{8},u_{9}\right\}
$ & $678$ & $15.25$ & $2.864$ & $3.002$\\
$8$ & $\left\{  u_{2},u_{3},u_{4},u_{5},u_{6},u_{7},u_{8},u_{9}\right\}  $ &
$667$ & $16.625$ & $2.978$ & $3.018$\\
$9$ & $\left\{  u_{3},u_{4},u_{5},u_{6},u_{7},u_{8},u_{9}\right\}  $ & $627$ &
$21.625$ & $3.086$ & $3.062$\\
$10$ & $\left\{  u_{1},u_{2},u_{3},u_{4},u_{5},u_{6},u_{7},u_{8}\right\}  $ &
$615$ & $23.125$ & $2.859$ & $2.960$\\
$11$ & $\left\{  u_{2},u_{3},u_{4},u_{5},u_{6},u_{7},u_{8}\right\}  $ & $604$
& $24.5$ & $2.908$ & $2.977$\\
$12$ & $\left\{  u_{4},u_{5},u_{6},u_{7},u_{8},u_{9}\right\}  $ & $567$ &
$29.125$ & $3.046$ & $3.016$\\
$13$ & $\left\{  u_{2},u_{3},u_{4},u_{5},u_{6},u_{7}\right\}  $ & $530$ &
$33.75$ & $2.908$ & $2.978$\\
$14$ & $\left\{  u_{2},u_{3},u_{4},u_{5},u_{10},u_{11},u_{12},u_{13}%
,u_{14}\right\}  $ & $443$ & $44.625$ & $3.018$ & $3.080$\\
$15$ & $\left\{  u_{1},u_{2},u_{3},u_{4},u_{5},u_{6}\right\}  $ & $441$ &
$44.875$ & $2.775$ & $2.894$\\
$16$ & $\left\{  u_{1},u_{2},u_{3},u_{4},u_{8},u_{9},u_{10},u_{11}%
,u_{15}\right\}  $ & $439$ & $45.125$ & $2.969$ & $3.048$\\
$17$ & $\left\{  u_{1},u_{2},u_{3},u_{4},u_{5},u_{11},u_{12},u_{13}%
,u_{14},u_{15},u_{16}\right\}  $ & $406$ & $50.75$ & $3.018$ & $3.031$\\
$18$ & $\left\{  u_{1},u_{2},u_{3},u_{4},u_{5}\right\}  $ & $337$ & $57.875$ &
$2.788$ & $2.931$\\
$19$ & $\left\{  u_{8},u_{9},u_{10},u_{11},u_{12},u_{13},u_{14},u_{15}%
,u_{16}\right\}  $ & $256$ & $67.625$ & $2.990$ & $3.212$\\
$20$ & $\left\{  u_{10},u_{11},u_{12},u_{13},u_{14},u_{15},u_{16}\right\}  $ &
$122$ & $84.75$ & $2.894$ & $2.822$\\\hline
\end{tabular}

\end{center}

The estimations from the two methods are also accurate in this case of gamma
distribution for the parameter $b.$ Here also the truncation proportion does
not affect the quality of estimations. When we consider the complete data, the
classical estimation is $\widehat{b}=3.04$ and the two new estimations are
$\widetilde{b}_{1}=3.018$ and $\widetilde{b}_{2}=3.054.$

As it was noticed in the examples above, the two methods lead to approximately
the same estimation results. Nevertheless, if the two estimations are
significantly different, it seems related to the quality of the selected data.
An important feature of this new approach is that the quality of estimations
is uninfluenced by the truncation proportion. The following section will give
further insights of the new approach.

\subsection{A Basic Feature of the New Approach}

The preceding results have shown the effectiveness of the new approach and
worked well in simulation experiments. Furthermore, the proposition below will
give an insight of a major feature of the new approach by considering the one
parameter exponential family. We will prove that for all truncation considered
formed by more than two points, from a sample of observations; if the ratios
of the relative frequencies of the $u_{i}$ are equal to the theoretical ones,
then we may obtain the true value of the parameter. We may conjecture that
when considering an arbitrary law of probability depending on $r$ parameters,
such that we have a truncation composed by $r+1$ points having exact empirical
ratios of the relative frequencies then we obtain the true values of the $r$ parameters.

\begin{proposition}
Consider a probability distribution $f$ from the one-parameter exponential
family with density%
\begin{equation}
f(x,\theta)=K(x)\times\exp[\theta T(x)-A(\theta)], \label{expo-fam}%
\end{equation}
where $\theta\in\mathbb{R}$ is the parameter, $T$ a statistic, $K(x)$ a
function of $x$ and $A$ is a function of the parameter $\theta.$ Assume that
we wish to estimate the parameter $\theta.$ If we consider a truncation having
two points\textbf{ }$x$\textbf{\ }and $y$\textbf{\ }with empirical
frequencies\textbf{ }$f_{1}$\ and $f_{2}$\textbf{\ }satisfying\textbf{ }%
$f_{1}/f_{2}=f(x,\theta)/f(y,\theta)$, then, using the approach considered
here, we obtain the true value of $\theta.$
\end{proposition}

\begin{proof}
1. If we consider the minimum of distance $d_{v}$ the result is
immediate.\newline2. Consider now the second method to estimate $m$. Consider
two values $x$ and $y$ from the exponential family with density given by
(\ref{expo-fam})$,$ with $\widetilde{\theta}$ being the estimation by the new
approach, and assume that their empirical frequencies $f_{1}$\ and $f_{2}$ are
such that%
\[
\frac{f_{1}}{f_{2}}=\frac{f(x,\widetilde{\theta})}{f(y,\widetilde{\theta})}.
\]
We obtain%
\[
\overline{u}=xf_{1}+yf_{2}=\frac{xK(x)\exp\left(  \widetilde{\theta
}T(x)\right)  +yK(y)\exp\left(  \widetilde{\theta}T(y)\right)  }%
{K(x)\exp\left(  \widetilde{\theta}T(x)\right)  +K(y)\exp\left(
\widetilde{\theta}T(y)\right)  }.
\]
Then, we solve on $\theta$ the following equation:%
\[
\left(  x-\frac{xK(x)\exp\left(  \widetilde{\theta}T(x)\right)  +yK(y)\exp
\left(  \widetilde{\theta}T(y)\right)  }{K(x)\exp\left(  \widetilde{\theta
}T(x)\right)  +K(y)\exp\left(  \widetilde{\theta}T(y)\right)  }\right)
K(x)\exp\left(  \theta T(x)\right)
\]%
\[
+\left(  y-\frac{xK(x)\exp\left(  \widetilde{\theta}T(x)\right)
+yK(y)\exp\left(  \widetilde{\theta}T(y)\right)  }{K(x)\exp\left(
\widetilde{\theta}T(x)\right)  +K(y)\exp\left(  \widetilde{\theta}T(y)\right)
}\right)  K(y)\exp\left(  \theta T(y)\right)  =0,
\]
after straightforward algebra we obtain%
\[
(x-y)\exp\left(  \widetilde{\theta}T(y)+\theta T(x)\right)  +(y-x)\exp\left(
\widetilde{\theta}T(x)+\theta T(y)\right)  =0,
\]
yielding the true value $\widetilde{\theta}=\theta.$ The proof is complete.
\end{proof}

\begin{remark}
Note that the frequencies $f_{1}$ and $f_{2}$ need not be exact, that is
$f_{1}$ may be different from $f(x,\theta)$ and also $f_{2},$ but we require
only that their ratio is equal to the theoretical one $f(x,\theta
)/f(y,\theta).\medskip$
\end{remark}

\textbf{Examples}\newline\textbf{Binomial distribution. }Consider again the
binomial distribution $\mathcal{B}(n,p)$ with $n=10$ and $p=0.3$ and assume
$n$ is known and we wish to estimate $p.$ Assume we have the following
truncation with only two points $\triangle=\left\{  0,1\right\}  .$ The exact
ratio of their probability distribution is given by $f(0,p)/f(1,p)=7/30,$
which is a rational value that will simplify the example. Choose the absolute
frequencies of the two values considered as being $n_{1}=7$ and $n_{2}=30$ for
the values $u_{1}=0$ and $u_{2}=1$ respectively, in order for having
$f_{1}/f_{2}=f(x,p)/f(y,p)=7/30.$ Using the first approach, that of the
minimum of distance $d_{v},$ we have to solve%
\[
\min_{p}d_{v}(\widehat{f},f)=\min_{p}\left[  \left\vert \frac{C_{10}%
^{0}(1-p)^{10}}{C_{10}^{1}p(1-p)^{9}}-\frac{7}{30}\right\vert +\left\vert
\frac{C_{10}^{1}p(1-p)^{9}}{C_{10}^{0}(1-p)^{10}}-\frac{30}{7}\right\vert
\right]  ,
\]
and we get the true value $\widetilde{p}_{1}=0.3$.

Using the second method we have to solve the following equation on $p$%
\[
\frac{0\times C_{10}^{0}(1-p)^{10}+1\times C_{10}^{1}p(1-p)^{9}}{C_{10}%
^{0}(1-p)^{10}+C_{10}^{1}p(1-p)^{9}}=\frac{30}{37},
\]
and we obtain also the exact result $\widetilde{p}_{2}=0.3.\medskip$%
\newline\textbf{Gamma distribution. }Consider the Gamma probability
distribution $G(a,b)$ with $a=10$ and $b=5.$ Assume that $a$ is known and we
wish to estimate $b.$ Consider the truncation $\triangle=\left\{  u_{3}%
,u_{8}\right\}  $ with $u_{3}=30.13$ and $u_{8}=60.02$. We have the following
value of the ratio $f\left(  u_{3},b\right)  /f\left(  u_{8},b\right)
\approx0.799$ (the result is an approximate result since for probability
density functions it is difficult to get an exact rational value but we will
show that the estimations are very close to the true value). Consider the
absolute frequencies $n_{3}=79.93$ (or $80$) and $n_{8}=100$ for the values
$u_{3}=30.13$ and $u_{8}=60.02$ respectively. We have then $n_{3}/n_{8}\approx
f\left(  u_{3},b\right)  /f\left(  u_{8},b\right)  .$ Using the minimum of
distance $d_{v},$ we have to solve%
\[
\min_{b}d_{v}(\widehat{f},f)=\min_{b}\left[  \left\vert
(79.93/100)-((30.13/60.02)^{9}\times\exp(-(1/b)\times
(30.13-60.02)))\right\vert \right.
\]%
\[
\left.  +(100/79.93)-((60.02/30.13)^{9}\times\exp(-(1/b)\times
(60.02-30.13)))\right]  ,
\]
and we get the result $\widetilde{b}_{1}\approx5.$

From the second method, we compute $\overline{u}=46.7438$ and solve on $b$ the
following equation%
\[
(30.13-46.7438)\times30.13^{9}\times\exp(-30.13/b)
\]%
\[
+(60.02-46.7438)\times60.02^{9}\times\exp(-60.02/b)=0.
\]
The result is $\widetilde{b}_{2}\approx5.$

Now assume that the parameters $a$ and $b$ are unknown and show how to jointly
estimate them using the new approach. Since now there are two unknown
parameters, we need to have three points from the support, so consider
$u_{1}=34.7702$, $u_{2}=57.5008$ and $u_{3}=74.5487$ with their corresponding
absolute frequencies $n_{1}=102,n_{2}=100$ and $n_{3}=34.$ We have to find $a$
and $b$ which minimize the distance $d_{v}$ that is $\min_{a,b}d_{v}%
(\widehat{f},f).$ The result is $\widetilde{a}\approx10.0454$ and
$\widetilde{b}\approx4.9739.$

\section{Elements of Comparison with the Classical Approach}

Our aim here is not to give a detailed comparison study which needs to be
investigated thoroughly, but only some elements of appreciation. A major
feature which characterizes this new approach from the others is that when we
have exact ratios of frequencies we obtain the true parameter and when their
difference from the theoretical ratios decrease the quality of estimation
increase even if we are using only a part from the sample of observations.
This is not the case for classical approaches. In classical approaches,
quality considerations are only viewed through mean properties of estimators
or their asymptotic behaviour. By combining the two proposed methods we have
in fact a point criterion. Another characteristics is that the proportion of
truncation has any effect on the quality of estimations. The first method uses
a well known method of minimum distance but with a new one which has an
important advantage of being symmetric, the property of which many traditional
distances do not have. However, the estimations are obtained in this case
implicitly so it is difficult to find explicit expressions and study their
properties to compare them with classical ones. Using the new distance we hope
having fast convergent estimators since we expect that the influence of the
errors in the frequencies will be slight in the new approach as we are using
ratios of frequencies. Consider now the second method of the new approach. We
use classical procedures of estimation such as the maximum likelihood
principle using the auxiliary distribution. We may obtain the estimators and
study their properties as commonly used and then preserves the advantages of
classical methods. In classical approach, given a sample, the estimation of
certain parameters such as the mean and variance do not change according to
the family of parent distributions. The latter information is not used and
this disadvantages the approach. However, in the new approach the estimations
of the mean and variance change according to the distribution from which the
data emanated.

The following two examples show the effectiveness of using the auxiliary
distribution.\medskip\newline\textbf{Example. }Consider the following
frequency table:

\begin{center}
\textbf{Table 7.\smallskip}%

\begin{tabular}
[c]{cccc}\hline
$x_{i}$ & $2$ & $3$ & $Total$\\\hline
$n_{i}$ & $n_{1}$ & $n_{2}$ & $n$\\
$\widehat{f}\left(  x_{i}\right)  =f_{i}$ & $f_{1}=\left(  n_{1}/n\right)  $ &
$f_{2}=\left(  n_{2}/n\right)  $ & $1$\\\hline
\end{tabular}

\end{center}

Any sample of observations that satisfies the preceding frequency table may
belong from one of the following distributions:%
\[
g_{1}\left(  x\right)  =\left\{
\begin{array}
[c]{c}%
\frac{x}{6}\\
0
\end{array}%
\begin{array}
[c]{c}%
if\text{ }x\in\left\{  1,2,3\right\}  ,\\
otherwise,
\end{array}
\right.  \text{ \ \ or \ \ }g_{2}\left(  x\right)  =\left\{
\begin{array}
[c]{c}%
\frac{x-1}{6}\\
0
\end{array}%
\begin{array}
[c]{c}%
if\text{ }x\in\left\{  2,3,4\right\}  ,\\
otherwise.
\end{array}
\right.
\]
The decision for determining which of the two distributions is more
appropriate for table 7, depends intuitively on the values $n_{1}$ and $n_{2}$
(or $f_{1}$ and $f_{2}$). However, if we use the classical maximum likelihood,
we obtain that the samples of observations were generated from distribution
$h_{1}$ whatever the values of $n_{1}$ and $n_{2},$ that is:%
\[
\left(  \frac{1}{6}\right)  ^{n_{1}}\times\left(  \frac{2}{6}\right)  ^{n_{2}%
}<\left(  \frac{2}{6}\right)  ^{n_{1}}\times\left(  \frac{3}{6}\right)
^{n_{2}}.
\]
We will show by using the new approach that the decision is more relevant.
Determine first the auxiliary distributions, $h_{1}$ and $h_{2}$, based on the
truncation $\triangle=\left\{  2,3\right\}  $, for $g_{1}$ and $g_{2}$
respectively. We obtain%
\[
h_{1}(x)=\left\{
\begin{array}
[c]{c}%
2/5\\
3/5\\
0
\end{array}%
\begin{array}
[c]{c}%
\text{ \ \ }if\text{ \ \ }x=2,\\
\text{ \ \ }if\text{ \ \ }x=3,\\
otherwise,
\end{array}
\right.  \text{ \ \ and \ \ }h_{2}(x)=\left\{
\begin{array}
[c]{c}%
1/3\\
2/3\\
0
\end{array}%
\begin{array}
[c]{c}%
\text{ \ \ }if\text{ \ \ }x=2,\\
\text{ \ \ }if\text{ \ \ }x=3,\\
otherwise.
\end{array}
\right.
\]
By using the maximum likelihood for $h_{1}$ and $h_{2}$, we have to decide
according to the quantities $\left(  2/5\right)  ^{n_{1}}\times\left(
3/5\right)  ^{n_{2}}$ and $\left(  1/3\right)  ^{n_{1}}\times\left(
2/3\right)  ^{n_{2}}.$ Solving the following inequality
\[
\left(  \frac{2}{5}\right)  ^{n_{1}}\times\left(  \frac{3}{5}\right)  ^{n_{2}%
}\leq\left(  \frac{1}{3}\right)  ^{n_{1}}\times\left(  \frac{2}{3}\right)
^{n_{2}},
\]
which is equivalent to $\left(  6/5\right)  ^{\alpha}\left(  9/10\right)
^{1-\alpha}\leq1,$ where $\alpha=n_{1}/n_{2},$ we obtain $0<\alpha\leq
-\log(9/10)/\log(4/3)=x_{0}\approx0.36624.$ If $0<\alpha<x_{0},$ the data were
generated from $g_{2}$ and if $x_{0}<\alpha<1$, the data were generated from
$g_{1}.$ We cannot make any decision about the case $\alpha=x_{0}.\medskip
$\newline\textbf{Example. }Consider a binomial distribution with parameters
$n=4$ and $p$ is unknown, from which we consider some samples of observations
of size 15 given in table 8 by their absolute frequencies and chosen in order
for having $\overline{x}=8/15.$

\begin{center}
\textbf{Table 8.\smallskip}%

\begin{tabular}
[c]{|c|ccccc|cc|}\hline
& \multicolumn{5}{|c|}{Values} & \multicolumn{2}{|c|}{}\\\cline{2-6}%
\cline{6-6}%
samples & $0$ & $1$ & $2$ & $3$ & $4$ & $\widehat{p}$ & $\widetilde{p}%
$\\\cline{1-6}\cline{2-8}%
$1$ & $7$ & $8$ & $0$ & $0$ & $0$ & $0.133$ & $0.222$\\
$2$ & $9$ & $5$ & $0$ & $1$ & $0$ & $0.133$ & $0.184$\\
$3$ & $9$ & $4$ & $2$ & $0$ & $0$ & $0.133$ & $0.139$\\
$4$ & $10$ & $3$ & $1$ & $1$ & $0$ & $0.133$ & $0.134$\\
$5$ & $10$ & $4$ & $0$ & $0$ & $1$ & $0.133$ & $0.216$\\
$6$ & $12$ & $0$ & $2$ & $0$ & $1$ & $0.133$ & $0.196$\\
$7$ & $13$ & $0$ & $0$ & $0$ & $2$ & $0.133$ & $0.385$\\\cline{1-6}\cline{3-8}%
\end{tabular}

\end{center}

It is clear that the information given by the samples are not the same,
nevertheless the classical estimation method gives us the same estimation
$\widehat{p}=8/(15\times4)\approx0.133.$ If we use the second method of the
new approach, we have to solve the following equation for each sample:%
\[
0\times h(0,p)+1\times h(1,p)+2\times h(2,p)+3\times h(3,p)+4\times
h(4,p)=\overline{x},
\]
where $h(x,p)$ is the corresponding auxiliary distribution. The estimations
given by the new method differ from sample to another as shown in the latest
column of table 8, which is natural since each sample provides a different
information about the parent distribution. We can also use the minimum of
distance $d_{v}$ and we get also the same conclusion.

\section{Perspectives for the New Approach}

\subsection{Model Selection From Truncated Data}

The fact that the distance $d_{v}$ is a metric allow to propose various
applications of this new measure. We can use it for model selection amongst
different probability families. We choose two or more possible candidate
parametric families of distributions, and for each alternative family,
estimate the parameters to select a specific candidate. Determine the distance
between the specific candidate and the empirical distribution using the new
metric $d_{v}.$ Finally, select the family which yields the minimum distance.
In view of the new approach this can also be done in case of truncated data as
opposed to classical approaches (see for example Cox \cite{Cox1},
\cite{Cox2}), Taylor and Jakeman \cite{Tay}) for model selection which can be
used, from the best of our knowledge, only for complete data.

To investigate this perspective thoroughly, samples of various sizes from
known distributions should be simulated, and the method for model selection
applied, we can score the selection as correct or not after repeating the
process a large number of times, the probability of correct selection could be
estimated according to a given sample size.

We can also use the new distance in cases where classical goodness of fit
tests cannot reject two candidate families. We can choose the one which yields
the minimum of distance $d_{v}.$

In the following examples, we shall select, in the first, between binomial
distributions from truncated data. In the second example, we select between a
Weibull and a Gamma distributions from right truncated data.\medskip
\newline\textbf{Selection from Binomial distributions. }We simulated $10000$
samples of size $100$ from a Binomial distribution $\mathcal{B(}%
8,0.1\mathcal{)}$ and each time we retained only the observations belonging
from $\{0,1,2,3\}$ with their frequencies. Then we tried to identify the law
simulated starting from the corresponding table of frequencies. We used the
distance $d_{v}$ to select between the original distribution of each simulated
sample and the distribution $\mathcal{B(}10,0.15\mathcal{)}$ and we score the
selection as correct if the distance between the empirical distribution and
the original one is less than with the alternative one $\mathcal{B(}%
15,0.15\mathcal{)}.$ The correct distribution was selected $98,8\%$.
Conversely, we simulated $10000$ samples of size $100$ from a Binomial
distribution $\mathcal{B(}10,0.15\mathcal{)}$ and we select with
$\mathcal{B(}8,0.1\mathcal{)}$, the correct distribution was selected
$99,43\%.$\medskip\newline\textbf{Selection between Weibull and Gamma
distributions. }We simulated $10000$ samples of size $1000$ from the weibull
distribution $\mathcal{W}(1.2,1.5)$ and we truncated them on right by
considering only observations above the cut-off $1.25$. Each truncated sample
was summarized into $11$ classes. We selected between $\mathcal{W}(1.2,1.5)$
and the Gamma distribution $G\left(  2,0.5\right)  .$ The distance $d_{v}$ has
selected the correct distribution, that is $\mathcal{W}(1.2,1.5),$ $98.16\%.$

We can also find, before selecting between distribution, the best fit from the
family of gamma distributions $G\left(  a,b\right)  $ of the truncated data
from a given probability density say $\mathcal{W}(1.2,1.5).$ We have then to
solve an optimization problem of finding the minimum of a function of two
variables, $\min_{a,b}d_{v}(\widehat{f},f)$ where $\widehat{f}$ is the
empirical distribution and $f\equiv G\left(  a,b\right)  $, using well known
methods such as Lavenberg-Marquardt using a computer algebra package. Also it
should be better to choose the number of bins for each truncated sample by an
optimal procedure, for example that of Birg\'{e} and Rozenholc \cite{Bir1}.

\subsection{Estimation of the initial trial value in \textit{EM }Algorithm}

The initial starting value is of great importance in convergence behaviour of
algorithms such as \textit{EM }Algorithm. Usually, as for the latter, the
initial trial value is guessed. Surprisingly, we will show that our procedure
gives an estimation of the starting value instead of having to guess. The
approach will be illustrated by the following classical example which was the
basis of the \textit{EM} algorithm.\medskip\newline\textbf{Example of Hartley
(1958) revisited. }Hartley \cite{Hart} used an algorithmic procedure to
estimate the parameter of a Poisson distribution from data on the pollution of
a sort of seeds by the presence of noxious weed seeds quoted from Snedecor
\cite{Sne} and truncated them by missing the frequencies of the values $0$ and
$1$ as shown in the following table 9 (Table 1 in Hartley \cite{Hart})\medskip

\begin{center}
\textbf{Table 9.\smallskip}%

\begin{tabular}
[c]{ccccccccccc}\hline
Values & missing & $0$ & $1$ &  &  &  &  &  &  & \\
& observed &  &  & $2$ & $3$ & $4$ & $5$ & $6$ & $7$ & $9$\\\hline
frequencies $n_{i}$ &  &  &  & $26$ & $16$ & $18$ & $9$ & $3$ & $5$ &
$1$\\\hline
\end{tabular}

\end{center}

Hartley \cite{Hart} has guessed the frequencies of the missing values $0$ and
$1$ by taking $n_{0}=4$ and $n_{1}=14,$ and after 4 steps of his algorithmic
procedure, which has been the basis of the well known EM algorithm for
incomplete data (Dempster, Laird and Rubin \cite{Demp}), has reached the
estimation $\widehat{\lambda}=3.026$ (see table 1 p.177 Hartley \cite{Hart}).
Using the second method, we get the estimation $\widetilde{\lambda}%
_{2}=3.1149.$ And by proportional allocation procedure we can see that the
frequencies we get are $n_{0}=4.29$ and $n_{1}=13.38$ which are close to the
guessed values. Using the distance $d_{v}$ we obtain the estimation
$\widetilde{\lambda}_{1}=3.8447,$ and by removing the last value which has a
small frequency $n_{7}=1,$ we obtain a better result $\widetilde{\lambda}%
_{1}=3.4441,$ which are also appreciable as starting values since in practice
the true parameter is unknown.\medskip\newline\textbf{Initial trial value for
mixture Normal Populations. }We shall present an application of the previous
method used for truncated data in the situation where we have a mixture
population of two normal distributions. In classical methods, we use the
merged distribution\textbf{ }$f=\alpha f_{1}+\left(  1-\alpha\right)  f_{2}%
$\textbf{ }and we estimate the parameters $\alpha,$ $m_{1}$ and $m_{2}%
$\textbf{ }using for example the EM algorithm which is based on maximizing the
complete likelihood of the merged distribution by an algorithmic procedure
from a guessed initial trial value. However, the problem of occurrence of
several local maxima is well-known for the setting of EM algorithm. Also,
Seidel, Mosler and Alker \cite{Sei} pointed out that the likelihood-ratio test
in mixture models depends on the choice of the initial trial value for the EM
algorithm. If the initial trial value is close to the true value it is clear
that the algorithm will converge in few steps to the true local maximum. We
will show that using the new approach we get an accurate estimated initial
trial value.

Assume we have a merged sample from two samples of observations of sizes
$n_{1}$ and $n_{2}$ from two normal distributions $f_{1}=N(m_{1},\sigma_{1})$
and $f_{2}=N(m_{2},\sigma_{2})$, with $m_{1}\neq m_{2}$. By assuming that
$\sigma_{1}$ and $\sigma_{2}$ are known, our aim is to estimate the means
$m_{1}$ and $m_{2},$ and also the merging proportion $\alpha$ of each population.

We will use a method based on truncations. The main idea being to split the
range of the merged sample into three suitably chosen parts. A central part
where the observations are highly merged, a left and right truncated parts
where the observations become mainly from one of the distributions considered.
If for example $m_{1}<m_{2}$, then to estimate $m_{1}$ we have to use the
chosen right truncated part (left truncation $\triangle$).

The procedure is summarized as follows:

\textbf{1.} We compute the sample mean $m_{g}$ of the merged observations.

\textbf{2.} For determining the location of the two means $m_{1}$ and $m_{2}$,
we compute the empirical standard deviation $S_{l}$ of the observations less
than $m_{g},$ and $S_{r}$ for those that are greater. Assume that $S_{l}%
<S_{r}$, in this case if $\sigma_{1}<\sigma_{2}$ then we deduce that $m_{1}$
is situated on the left of $m_{g}$. Otherwise, it will be assumed to be on its
right. We follow the same idea for the case $S_{l}>S_{r}.$ If $\sigma
_{1}=\sigma_{2}$ we pass directly to the third step.

\textbf{3.} Assume that $m_{1}$ is on the left. It is well known that for a
normal distribution $N\left(  m,\sigma\right)  $ we have $P(\left]
m-\sigma,m+\sigma\right[  )$ $\simeq0.68.$ We hope that on the left of
$\sup_{l}=m_{g}-\sigma_{2}$ the number of observations generated from
$N\left(  m_{2},\sigma_{2}\right)  $ is negligible, and on the right of
$\min_{r}=m_{g}+\sigma_{1}$ the number of observations generated from
$N(m_{1},\sigma_{1})$ is also negligible. Hence, to estimate $m_{1}$, we
consider only the part of observations situated on the left of $m_{g}%
-\sigma_{2}$, and to estimate $m_{2}$ we consider the part situated on the
right of $m_{g}+\sigma_{1}.$

The following example will provide some feel for the accuracy of the
procedure.\newline\textbf{Example. }We consider the case where $\sigma
_{1}=\sigma_{2}.$ consider two samples of observations generated from
$N(m_{1},\sigma_{1})$ and $N(m_{2},\sigma_{2}),$ where $m_{1}=1.3$ and
$m_{2}=2.4$, with known $\sigma_{1}=\sigma_{2}=1$ and sizes $n_{1}=300$ and
$n_{2}=200.$ We combine them to obtain a merged sample of size $n=500.$ We
have chosen the distributions in such a way that the histogram (Fig.1) of the
merged sample does not show directly the existence of a mixture of two
distributions.\ When\ the histogram of the merged population is bimodal the
situation is more easier, since when taking a suitably left (or right) part we
get more accurate estimation from the situation that this part will have a
negligible number of observations from the second distribution.

\begin{center}%
\begin{center}
\includegraphics[
height=3.1574in,
width=3.9271in
]%
{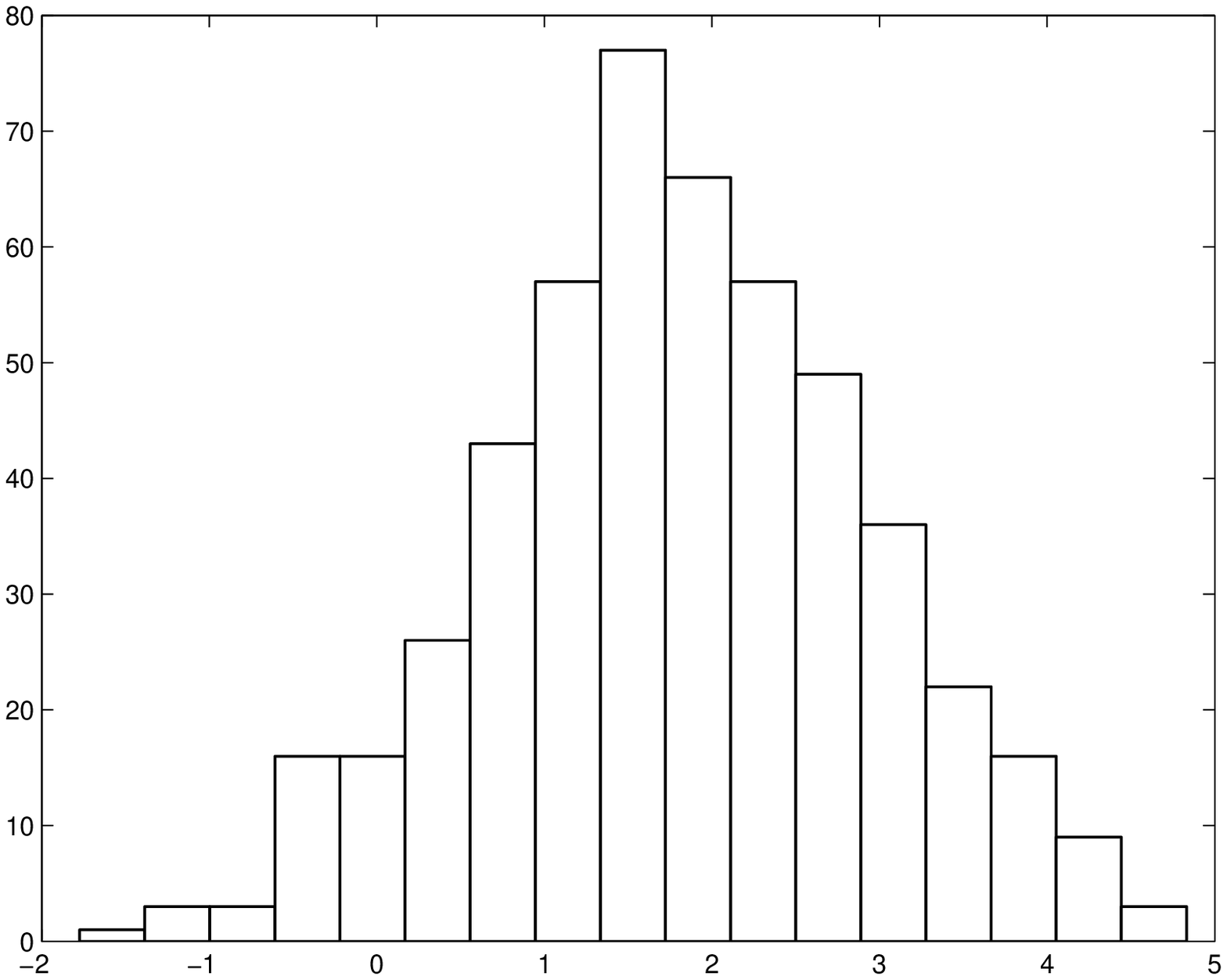}%
\end{center}

Fig 1. Merged histogram of two normal distributions $N(1.3,1)$ and $N(2.4,1).$
\end{center}

It should be stressed that the histogram is one modal and does not show at
first glance any mixture situation. Following the steps of the procedure we
begin by calculating the mean of the resulting merged sample and we obtain
$m_{g}=1.8046$. Since the standard deviations are assumed to be equal then we
compute directly $\sup_{l}=m_{g}-\sigma_{2}=0.8046.$ By grouping the
observations on the left of $\sup_{l}$ (which constitute the chosen right
truncated part) in $7$ classes we obtain the following table:\medskip

\textbf{Table 10.}\smallskip%

\begin{tabular}
[c]{ccccccc}\hline
$u_{i}$ & $-1.5589$ & $-1.1294$ & $-0.6998$ & $-0.2703$ & $-0.1593$ &
$0.5888$\\\hline
$n_{i}$ & $1$ & $3$ & $6$ & $17$ & $24$ & $41$\\\hline
\end{tabular}
\medskip\newline

Using the distance $d_{v}$ we obtain for all the truncation $\widetilde{m}%
_{1}^{(dv)}=1.244$ and by deleting $u_{1}$ we get the value $\widetilde{m}%
_{1}^{(dv)}=1.2516.$

The sample mean of the observations on the left of $\sup_{l}$ is given by
$\overline{u}_{l}=0.1483.$ Using the second method we have to solve on $m$ the
following formula%
\begin{equation}
\frac{u_{1}\times\exp\left[  \frac{-\left(  u_{1}-m\right)  ^{2}}{2\sigma^{2}%
}\right]  +u_{2}\times\exp\left[  \frac{-\left(  u_{2}-m\right)  ^{2}}%
{2\sigma^{2}}\right]  +...+u_{k}\times\exp\left[  \frac{-\left(
u_{k}-m\right)  ^{2}}{2\sigma^{2}}\right]  }{\exp\left[  \frac{-\left(
u_{1}-m\right)  ^{2}}{2\sigma^{2}}\right]  +\exp\left[  \frac{-\left(
u_{2}-m\right)  ^{2}}{2\sigma^{2}}\right]  +...+\exp\left[  \frac{-\left(
u_{k}-m\right)  ^{2}}{2\sigma^{2}}\right]  }=\overline{u}_{l}. \label{m2}%
\end{equation}
we obtain the estimation $\widetilde{m}_{1}=1.2646.$ By deleting the first
value $u_{1}$ which has a weak frequency $n_{1}=1,$ that is using the
truncation $\bigtriangleup=\left\{  u_{2},u_{3},u_{4},u_{5},u_{6}\right\}  ,$
(we compute again $\overline{u}_{l}=0.1734$) we obtain a better estimation
$\widetilde{m}_{1}=1.3011,$ which is very close to the true value $m_{1}=1.3$.

To estimate $m_{2},$ we consider the part situated on the right of $\min
_{r}=m_{g}+\sigma_{1}=2.8046.$ Grouping the observations on the right of
$\inf_{d}$ (which constitute the chosen right part) in $7$ classes we obtain
the following table:\medskip

\textbf{Table 11.}\smallskip%

\begin{tabular}
[c]{ccccccc}\hline
$u_{i}$ & $2.979$ & $3.316$ & $3.653$ & $3.990$ & $4.326$ & $4.663$\\\hline
$n_{i}$ & $38$ & $25$ & $15$ & $9$ & $7$ & $3$\\\hline
\end{tabular}
\medskip

Using the distance $d_{v}$ for all the truncation we get $\widetilde{m}%
_{2}^{(dv)}=2.397.$ The sample mean of the observations on the right part is
given by $\overline{u}_{d}=3.523.$ Using formula (\ref{m2}) with $\overline
{u}_{d}$, we obtain the result $\widetilde{m}_{2}=2.245.$ Deleting the extreme
values $u_{1}$ and $u_{6}$ we obtain $\widetilde{m}_{2}=2.412.$

The mixture proportion $\alpha$ can easily be estimated using the formula
$\alpha\times\widetilde{m}_{1}+\left(  1-\alpha\right)  \times\widetilde
{m}_{2}=m_{g}.$

Considering the estimations obtained, which are close to the true values of
$m_{1}$ and $m_{2},$ it is clear that the EM algorithm will converge fastly to
the unique solutions.

\subsection{Test of Goodness of Fit Based on the New Distance}

We can obtain empirical quantile estimations of $d_{v}$ using Montecarlo or
Bootstrapping technics, and use them in a test of goodness of fit for a
specified probability distribution. We simulate $N$ samples of the same size
from the specified probability distribution and calculate the distances
$d_{v}^{(1)},...,d_{v}^{(N)}.$ We can then estimate the asymptotic
distribution of $d_{v}$ by%
\begin{equation}
F_{d_{v}}(d)=\frac{\#d_{v}^{(i)}<d}{N}.
\end{equation}
Consequently, for a sample of the same size we compute $d_{v}^{(obs)}$ and we
reject the hypothesis that it belongs from the specified distribution if
$F_{d_{v}}(d_{v}^{(obs)})>(1-\alpha)$ for a given level of significance
$\alpha.$

The values $d_{v}^{(1)},...,d_{v}^{(N)}$ may be obtained from the empirical
distribution function $F_{n}$ of the sample.

\subsection{Quality of Data}

The fact that the new measure $d_{v}$ is not equivalent to classical ones
means that it treats other aspects not investigated by the latter. This may
open new perspectives such as making decision about the accuracy of an
estimation in cases where the classical and new estimations are close to each
others. In cases where the classical estimation and the new one using $d_{v}$
are significantly different then we can say that the sample of observations
considered does not restore coherently all necessary information about the
parent distribution from which it emanated.

\section{Concluding Remarks}

In the foregoing study, we have presented a new statistical point estimation
method which found be useful in truncated and grouped and censored data
situations. A new distance between probability distributions was introduced.
It measures the difference between the variations of two given probability
distributions. We introduced an auxiliary distribution based on a truncation,
from a chosen family of probability distributions. This new distribution will
have the same parameters to estimate as the parent one. We use then
statistical methods to estimate the parameters of the random variable under
study using the empirical and new auxiliary distribution in the region that
captures the data, from which we determine the corresponding parent
distribution. The later is the estimation by the new method. Using the new
distance introduced we also estimate by the minimum distance approach and use
the resulting estimation as a control on the accuracy of estimation obtained
by the former method. We have obtained a result which states that if we have
to estimate the parameter of a probability distribution from the one parameter
exponential family, then it suffices to have two points with exact ratio of
frequencies, that is equal to the theoretical one expressed by the ratio of
the value of the probability distribution on these two points, to obtain the
true value of the parameter. We have conjectured that if we have in general
$r$ parameters, then it suffices to have $r+1$ points with exact ratios of
their frequencies to obtain the $r$ true parameters exactly. The later result
need to be proved rigorously in a general setting for other distributions than
the class considered. A large comparative study between the classical and new
methods should also be investigated. We presented some perspectives of the new
approach such as model selection from truncated data using the new distance,
estimation of the first trial value in the celebrate \textit{EM} algorithm in
the case of truncation and for mixture of two normal populations, a test of
goodness of fit based on the new distance, decision making about the quality
of estimations and data.


\begin{thebibliography}{99}                                                                                               %


\bibitem {Ame}Amemiya, T. (1985). \textit{Advanced Econometrics. }Cambridge:
Harvard University Press.

\bibitem {Bir1}Birg\'{e}, L. and Rozenholc,Y. (2006) How many bins should be
put in a regular histogram. ESAIM: Probability and Statistics,Vol. 10, p. 24-45.

\bibitem {Cox1}Cox, D. R. (1961). Tests of separate families of hypotheses. In
\textit{Proceedings of the Fourth Berkeley Symposium, }Vol. 1, 105-123.
Berkeley: University of California Press.

\bibitem {Cox2}Cox, D. R. (1962). Further results on tests of separate
families of hypotheses. \textit{J. R. Statist. Soc.\ }B \textit{n}$%
{{}^\circ}%
$ 24, \textit{pp} 406--424.

\bibitem {Demp}Dempster, A. P., Laird, N. M. and Rubin, D. B. (1977) Maximum
likelihood from incomplete data via the EM algorithm, \textit{J. R. Statist.
Soc.\ }B \textit{n}$%
{{}^\circ}%
$ 1, \textit{pp} 1-38.

\bibitem {Efr}Efron, B. and Petrosian, V. (1999). Nonparametric methods for
doubly truncated data, \textit{J. Am. Stat. Assoc. }Vol. 94. No. 447. pp. 824-834.

\bibitem {Hart}Hartley, H. O. (1958). Maximum likelihood estimation from
incomplete data, \textit{Biometrics}, June , pp 174-194.

\bibitem {Kap}Kaplan, E. L. and Meier, P. (1958). Nonparametric estimation
from incomplete observations. \textit{J. Am. Stat. Assoc. }Vol. 53. pp. 457-481.

\bibitem {Kle}Klein, J. P. and Zhang, M. J. (1996). Statistical challenges in
comparing chemotherapy and bone marrow transplantation as a treatment for
leukemia, \textit{Lifetime Data: Models in Reliability and Survival Analysis,
N.P. Jewel, 175-185.}

\bibitem {Leh}Lehmann, E. L. \& Casella, G. (1998). \textit{Theory of point
estimation}. Springer, New York.

\bibitem {Lyn}Lynden-Bell, D. (1971). A method of allowing for known
observational selection in small samples applied to 3CR quasars. \textit{Mon.
Not. R. Astr. Soc. }Vol.155. pp. 95-118.

\bibitem {Par}Parzen, E. (1962). On estimation of a probability density
function and mode. \textit{Ann. Math. Stat., }1065-1076.

\bibitem {Sha}Shaw, D. (1988). On-Site samples regression problems of
nonnegative integers, truncation, and endogenous stratification.
\textit{Journal of Econometrics, }37, pp. 211-223.

\bibitem {Sei}Seidel, W., Mosler, K., and Alker, M. (2000). A cautionary note
on likelihood ratio tests in mixture models. \textit{Ann. Ist. Stat. Math.,
}52, 481-487.

\bibitem {Sne}Snedecor, G. W. (1956). \textit{Statistical Methods}. 5th ed.,
The Iowa State College Press, Ames, Iowa.

\bibitem {Tay}Taylor, J. A. and Jakeman, A. J. (1985). Identification of a
distributional model. \textit{Commun. Statist.- Simula. Computa., }14(2), 497-508.
\end{thebibliography}
\end{document}